\documentclass[aps,prl,twocolumn,superscriptaddress,longbibliography]{revtex4-2}

\usepackage{amsmath,amssymb,amsthm,mathtools,bm}
\usepackage{graphicx}
\usepackage[hidelinks]{hyperref}

\newtheorem{theorem}{Theorem}

\newtheorem{lemma}{Lemma}

\newcommand{\tr}{\mathrm{Tr}}
\newcommand{\id}{\mathbf{1}}

\newcommand{\half}{\tfrac{1}{2}}

\newcommand{\Jq}{J_{\mathrm{q}}}
\newcommand{\cgeom}{c_{\mathrm{geom}}}

\newcommand{\sympl}{\sigma}
\newcommand{\cov}{\Gamma}
\newcommand{\drift}{K}
\newcommand{\diffu}{D}
\newcommand{\CPmat}{M}
\newcommand{\Bures}{\mathrm{Bures}}
\newcommand{\BKM}{\mathrm{BKM}}

\begin{document}

\title{Score Reversal Is Not Free for Quantum Diffusion Models}

\author{Ammar Fayad}
\affiliation{Massachusetts Institute of Technology, Cambridge, Massachusetts 02139, USA}

\begin{abstract}
Diffusion-based generative modeling suggests reversing a noising semigroup by adding a score drift. For continuous-variable Gaussian Markov dynamics, complete positivity couples drift and diffusion at the generator level. For a quantum-limited attenuator with thermal parameter $\nu$ and squeezing $r$, the fixed-diffusion Wigner-score (Bayes) reverse drift violates CP iff $\cosh(2r)>\nu$. Any Gaussian CP repair must inject extra diffusion, implying $-2\ln F\ge \cgeom(\nu_{\min})\,\mathcal{I}_{\mathrm{dec}}^{\mathrm{wc}}$.
\end{abstract}

\maketitle

\paragraph{Introduction.}
Time reversal of Markov dynamics is ``free'' classically: for linear-Gaussian diffusions, the Bayes reverse drift produces a valid reverse diffusion at fixed $D^{\mathrm{cl}}\succeq 0$~\cite{Anderson1982,Haussmann1986}.
Quantum diffusion proposals frequently follow this logic by extracting a reverse drift from the Wigner--Fokker--Planck equation and lifting it to a Gaussian quantum reverse semigroup~\cite{Nasu2025}.
However, for continuous-variable Gaussian Markov dynamics, complete positivity imposes the generator constraint
$M=D+i(K\sigma+\sigma K^T)\succeq 0$~\cite{HHW2010}:
drift and diffusion cannot be specified independently.

We show a sharp quantum obstruction to the fixed-diffusion score-lift.
For the one-mode quantum-limited attenuator and a squeezed-thermal running reference $(\nu,r)$, the Bayes/Wigner score term drives $M$ nonpositive exactly when $\cosh(2r)>\nu$ (Theorem~\ref{thm:nogo}), a phase boundary with no classical analogue.
Enforcing CP within the Gaussian family therefore requires additional diffusion; combining the Gaussian de~Bruijn identity~\cite{Toscano2021} with Petz monotone-metric comparison~\cite{Petz1996,mosonyi2011quantum} yields an operational worst-case fidelity floor (Theorem~\ref{thm:floor}).

This Letter diagnoses the failure mode of the \emph{fixed-diffusion Gaussian semigroup score-lift}, i.e.\ the direct continuous-variable analogue of classical score reversal; in particular, the Bayes/score reverse drift at fixed diffusion is written explicitly in reverse-diffusion form in Ref.~\cite{Nasu2025}.
Our results do not contradict Petz recovery, which is CP by construction~\cite{Kwon2022}, nor do they preclude non-Gaussian/measurement-based approaches (e.g.\ Ref.~\cite{liu2025measurement}) that leave the fixed-diffusion Gaussian semigroup ansatz; rather, they provide a sharp obstruction and a quantitative benchmark \emph{within} the Gaussian reverse-semigroup class.\footnote{Reverse-time sign conventions and invariance of the CP spectrum under $K\mapsto -K$ are given in Supplemental Material~\cite{SM}.}

\paragraph{Gaussian quantum channels and complete positivity.}
We use the Heinosaari--Holevo--Wolf (HHW) conventions~\cite{HHW2010}: canonical operators satisfy $[Q_j,P_k]=i\delta_{jk}$, the symplectic form is
$\sigma=\bigoplus_j\!\left(\begin{smallmatrix}0&1\\-1&0 \end{smallmatrix}\right)$, and covariance matrices satisfy
$\Gamma+i\sigma\succeq 0$ (vacuum symplectic eigenvalues $\nu_k=1$).

A centered Gaussian channel acts on Weyl operators as
$W_\xi \mapsto W_{X^T\xi}\exp(-\half\,\xi^TY\xi)$.
Complete positivity is equivalent to~\cite{HHW2010}
\begin{equation}\label{eq:CP}
Y \succeq i(\sigma - X\sigma X^T).
\end{equation}
For a Gaussian dynamical semigroup with covariance generator
$\dot{\Gamma}_t = K_t\Gamma_t + \Gamma_t K_t^T + D_t$ (infinitesimally, $X=\id + K\,dt$, $Y=D\,dt$),
expanding~\eqref{eq:CP} to first order yields the \emph{generator CP matrix}
\begin{equation}\label{eq:genCP}
M_t := D_t + i(K_t\sigma + \sigma K_t^T) \succeq 0.
\end{equation}
This is the uncertainty principle at the level of Gaussian generators: drift cannot be chosen independently of diffusion.

\paragraph{Classical score reversal is free; quantum is not.}
For linear-Gaussian GKLS generators (quadratic Hamiltonian, linear Lindblad operators), the Wigner function evolves by a Fokker--Planck equation~\cite{gardiner2004quantum}
\begin{equation}\label{eq:FP}
\partial_s W_s = -\nabla\!\cdot\!(K_s x\, W_s) + \half\nabla^T\!\big(D_s^{\mathrm{cl}}\nabla W_s\big).
\end{equation}
For Gaussian marginals $W_s\propto \exp(-\half x^T\Sigma_s^{-1}x)$, classical time reversal gives the Bayes reverse drift~\cite{Anderson1982,Haussmann1986}
\begin{equation}\label{eq:revdrift}
K^{\mathrm{Bayes}} = K^{\mathrm{fwd}} + D^{\mathrm{fwd}}\Gamma_{\tau_s}^{-1},\qquad
D^{\mathrm{Bayes}} = D^{\mathrm{fwd}},
\end{equation}
where $\Gamma_{\tau_s}$ is the covariance of the running Gaussian reference in HHW units.
Classically, validity requires only $D^{\mathrm{cl}}\succeq 0$.
Quantum mechanically, the lifted generator must additionally satisfy~\eqref{eq:genCP}, which can fail.

\paragraph{No-go theorem.}
We specialize to the one-mode phase-covariant quantum-limited attenuator,
\begin{equation}
K^{\mathrm{fwd}} = -\gamma\id_2,\qquad D^{\mathrm{fwd}} = 2\gamma\id_2,
\end{equation}
and take a squeezed-thermal running reference
$\Gamma_{\tau_s}=\mathrm{diag}(\nu e^{2r},\nu e^{-2r})$ with $\nu\ge 1$, $r\ge 0$ (HHW vacuum $\nu=1$).

\begin{theorem}[No-go for noiseless quantum score reversal]\label{thm:nogo}
The Bayes reverse candidate~\eqref{eq:revdrift} has generator CP matrix
\begin{equation}
M^{\mathrm{Bayes}}
= \underbrace{2\gamma(\id_2 - i\sigma)}_{\displaystyle M^{\mathrm{fwd}}\succeq 0}
+ \underbrace{2i\gamma\!\left(\Gamma_{\tau_s}^{-1}\sigma + \sigma\Gamma_{\tau_s}^{-1}\right)}_{\displaystyle \Delta M\ \text{(Hermitian, traceless)}}.
\end{equation}
Let $u = \tfrac{1}{\sqrt{2}}(1,i)^T$ (so $\sigma u = iu$). Then
\begin{equation}\label{eq:threshold}
\boxed{u^\dagger M^{\mathrm{Bayes}} u = 4\gamma\!\left(1 - \frac{\cosh(2r)}{\nu}\right).}
\end{equation}
Consequently, $M^{\mathrm{Bayes}}\not\succeq 0$ (CP violation) \emph{iff}
\begin{equation}\label{eq:nogo}
\cosh(2r)>\nu.
\end{equation}
In particular, for unsqueezed thermal references ($r=0$) the obstruction never activates for $\nu\ge 1$; squeezing is the source of CP failure in this phase-covariant one-mode setting.
\end{theorem}

\emph{Proof.}
Substituting $K^{\mathrm{Bayes}} = -\gamma\id_2 + 2\gamma\Gamma_{\tau_s}^{-1}$ into~\eqref{eq:genCP} gives
$M^{\mathrm{Bayes}}=2\gamma(\id_2-i\sigma)+2i\gamma(\Gamma_{\tau_s}^{-1}\sigma+\sigma\Gamma_{\tau_s}^{-1})$.
For $u=\tfrac{1}{\sqrt{2}}(1,i)^T$ with $\sigma u=iu$,
\[
u^\dagger M^{\mathrm{Bayes}}u
=4\gamma\!\left(1-u^\dagger\Gamma_{\tau_s}^{-1}u\right)
=4\gamma\!\left(1-\frac{\cosh(2r)}{\nu}\right),
\]
yielding~\eqref{eq:threshold}.
Moreover, since $\Gamma_{\tau_s}^{-1}=\mathrm{diag}(a,b)$ is diagonal in the one-mode case,
a direct $2\times 2$ computation gives
$\Gamma_{\tau_s}^{-1}\sigma+\sigma\Gamma_{\tau_s}^{-1}=(a+b)\sigma
= (\tr\,\Gamma_{\tau_s}^{-1})\,\sigma$ with $\tr(\Gamma_{\tau_s}^{-1})=2\cosh(2r)/\nu$.
Thus $M^{\mathrm{Bayes}}$ is a real linear combination of $\id_2$ and $i\sigma$:
\[
M^{\mathrm{Bayes}}=2\gamma\,\id_2+\alpha\, i\sigma,\qquad
\alpha:=2\gamma\,\tr(\Gamma_{\tau_s}^{-1})-2\gamma.
\]
Since $i\sigma$ has eigenvalues $\pm 1$ with eigenvectors $u,u^*$, the eigenvalues of $M^{\mathrm{Bayes}}$ are
\[
\lambda_{+}=2\gamma+\alpha=2\gamma\,\tr(\Gamma_{\tau_s}^{-1})
=4\gamma\,\frac{\cosh(2r)}{\nu}\ge 0,
\]
\[
\lambda_{-}=2\gamma-\alpha=4\gamma-2\gamma\,\tr(\Gamma_{\tau_s}^{-1})
=4\gamma\Bigl(1-\frac{\cosh(2r)}{\nu}\Bigr).
\]
Hence $M^{\mathrm{Bayes}}\succeq 0$ iff $\lambda_{-}\ge 0$, i.e.\ $\cosh(2r)\le \nu$. \hfill$\square$

\paragraph{Physical picture.}
Write $M^{\mathrm{Bayes}}=M^{\mathrm{fwd}}+\Delta M$ where $\Delta M$ is Hermitian and traceless.
The forward matrix $M^{\mathrm{fwd}}=2\gamma(\id_2-i\sigma)$ has eigenvalues $\{0,\,4\gamma\}$: one direction sits exactly at the CP boundary, while the other has a $4\gamma$ ``buffer.''
Because $\Delta M$ is traceless, it redistributes the spectrum: it raises one eigenvalue while lowering the other.
As squeezing increases, the lowered eigenvalue crosses below zero precisely when $\cosh(2r)>\nu$, producing a sharp CP phase boundary.

\paragraph{The quantum noise floor.}
To restore CP within the Gaussian family, a repaired reverse decoder must add diffusion
\begin{equation}
D^{\mathrm{rev}} = D^{\mathrm{Bayes}} + \Delta D_{\mathrm{qu}},\quad
\Delta D_{\mathrm{qu}}\succeq 0,\quad
M^{\mathrm{Bayes}} + \Delta D_{\mathrm{qu}} \succeq 0.
\end{equation}
The minimal repair is a pointwise semidefinite program (SDP),
\begin{equation}\label{eq:SDP}
\Delta D_{\mathrm{qu}}^\star(s)
=\arg\min_{\substack{\Delta D=\Delta D^T\succeq 0\\ M^{\mathrm{Bayes}}(s)+\Delta D\succeq 0}}
\tr\!\big(\Delta D\,\Jq[\tau_s]\big),
\end{equation}
weighted by the BKM quantum Fisher information~\cite{Toscano2021}.
The Gaussian quantum de~Bruijn identity~\cite{Toscano2021} quantifies the entropy-rate increment from added diffusion:
\begin{equation}\label{eq:deBruijn}
\Delta\!\left(\frac{d}{ds}S(\rho_s)\right)
= \half\,\tr\!\big(\Delta D_{\mathrm{qu}}(s)\,\Jq[\rho_s]\big)\ \ge 0.
\end{equation}
For a prescribed family of target states $\mathcal{C}$, define the worst-case (wc) decoder irreversibility
$\mathcal{I}_{\mathrm{dec}}^{\mathrm{wc}}(S):=\sup_{\rho_0\in\mathcal{C}}\int_0^S \half\,\tr(\Delta D_{\mathrm{qu}}(s)\Jq[\rho_s])\,ds$.
Full definitions and conventions are in Supplemental Material~\cite{SM}.

\begin{theorem}[Quantum noise floor for Gaussian decoders]\label{thm:floor}
For phase-covariant $n$-mode Gaussian reverse decoders whose repaired trajectory has a uniform gap from purity,
$\nu_k(s)\ge \nu_{\min}>1$ (HHW units), one has the operational fidelity bound
\begin{equation}\label{eq:noisefloor}
\boxed{
\sup_{\rho_0\in\mathcal{C}}\!\big[-2\ln F(\rho_0,\hat{\rho}_S^{\,\mathrm{Gauss}}(\rho_0))\big]
\ \ge\
\cgeom(\nu_{\min})\,\mathcal{I}_{\mathrm{dec}}^{\mathrm{wc}}(S),
}
\end{equation}
where the geometric constant from Petz monotone-metric comparison~\cite{Petz1996,mosonyi2011quantum} is
\begin{equation}\label{eq:cgeom_main}
\cgeom(\nu)=\frac{1}{2\nu\ln\!\frac{\nu+1}{\nu-1}}\ \in\ (0,\tfrac{1}{4}].
\end{equation}
\end{theorem}

\emph{Proof sketch.}
Equation~\eqref{eq:deBruijn} measures entropy production induced by CP-enforcing diffusion in the BKM metric.
Under a uniform gap from purity $\nu_{\min}>1$, Petz monotone-metric comparison yields a state-independent bound
$g^{\mathrm{Bures}}\succeq \cgeom(\nu_{\min})\,g^{\mathrm{BKM}}$ on the relevant thermal spectra, which converts BKM cost into fidelity loss and integrates to~\eqref{eq:noisefloor}. Full derivations are in Supplemental Material~\cite{SM}. \hfill$\square$

\paragraph{Numerical validation.}
Figure~\ref{fig:CP} validates Theorem~\ref{thm:nogo} by mapping $\lambda_{\min}(M^{\mathrm{Bayes}})$ across the $(\nu,r)$ plane and confirming the sharp boundary $\cosh(2r)=\nu$ (HHW vacuum $\nu=1$).
Figure~\ref{fig:noisefloor} validates Theorem~\ref{thm:floor} by comparing worst-case infidelity $-2\ln F$ to the lower bound $\cgeom(\nu_{\min})\,\mathcal{I}_{\mathrm{dec}}^{\mathrm{wc}}$ versus depth $S$.
Verification checks are given in Supplemental Material~\cite{SM}.

\begin{figure}[t]
\centering
\includegraphics[width=\columnwidth]{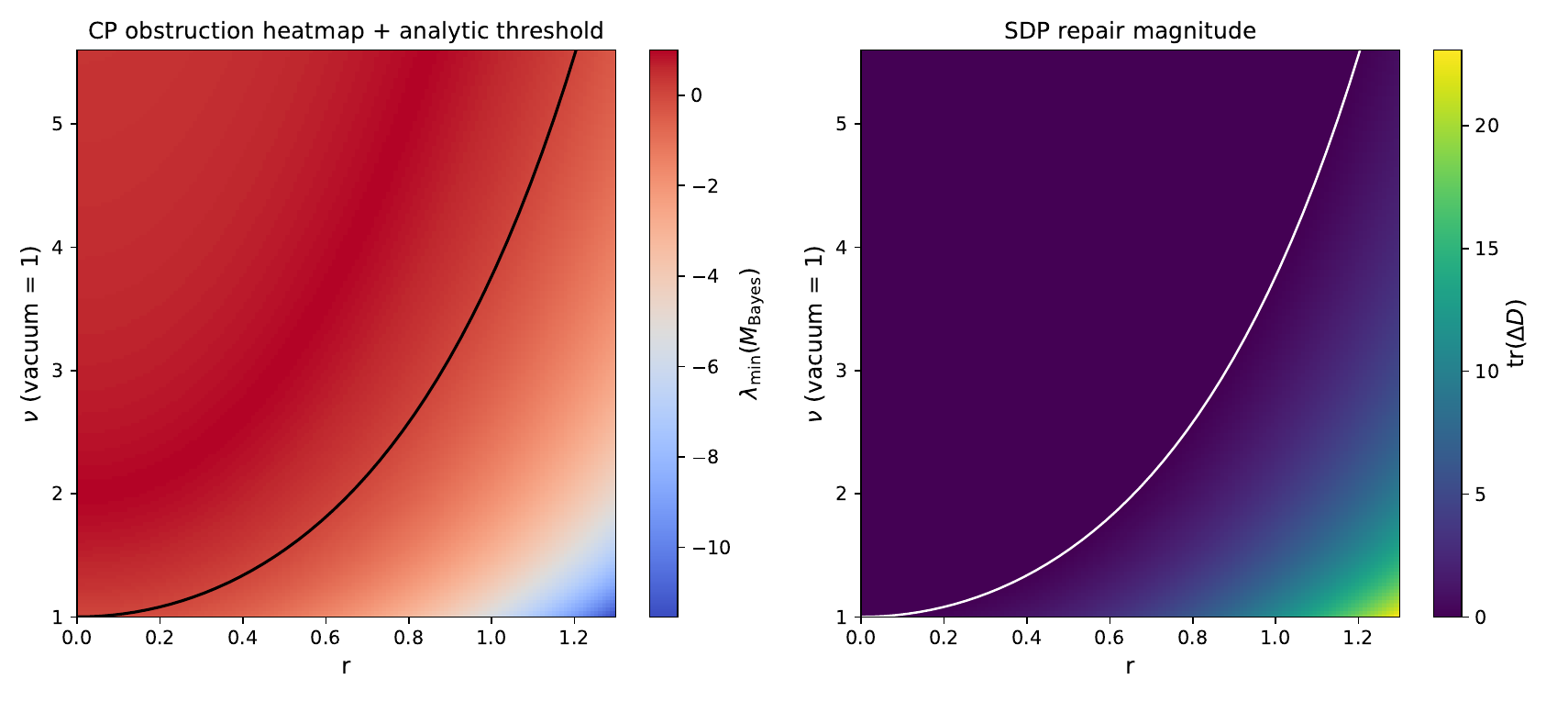}
\caption{No-go phase diagram for the fixed-diffusion score-lift (HHW conventions, vacuum $\nu=1$).
\emph{Left:} minimum eigenvalue $\lambda_{\min}(M^{\mathrm{Bayes}})$ over $(\nu,r)$; CP violation occurs exactly where $\lambda_{\min}<0$.
The analytic threshold curve $\cosh(2r)=\nu$ is overlaid.
\emph{Right:} magnitude $\tr(\Delta D_{\mathrm{qu}}^\star)$ of the minimal CP repair diffusion (SDP), showing mandatory noise injection beyond the CP boundary.}
\label{fig:CP}
\end{figure}

\begin{figure}[t]
\centering
\includegraphics[width=\columnwidth]{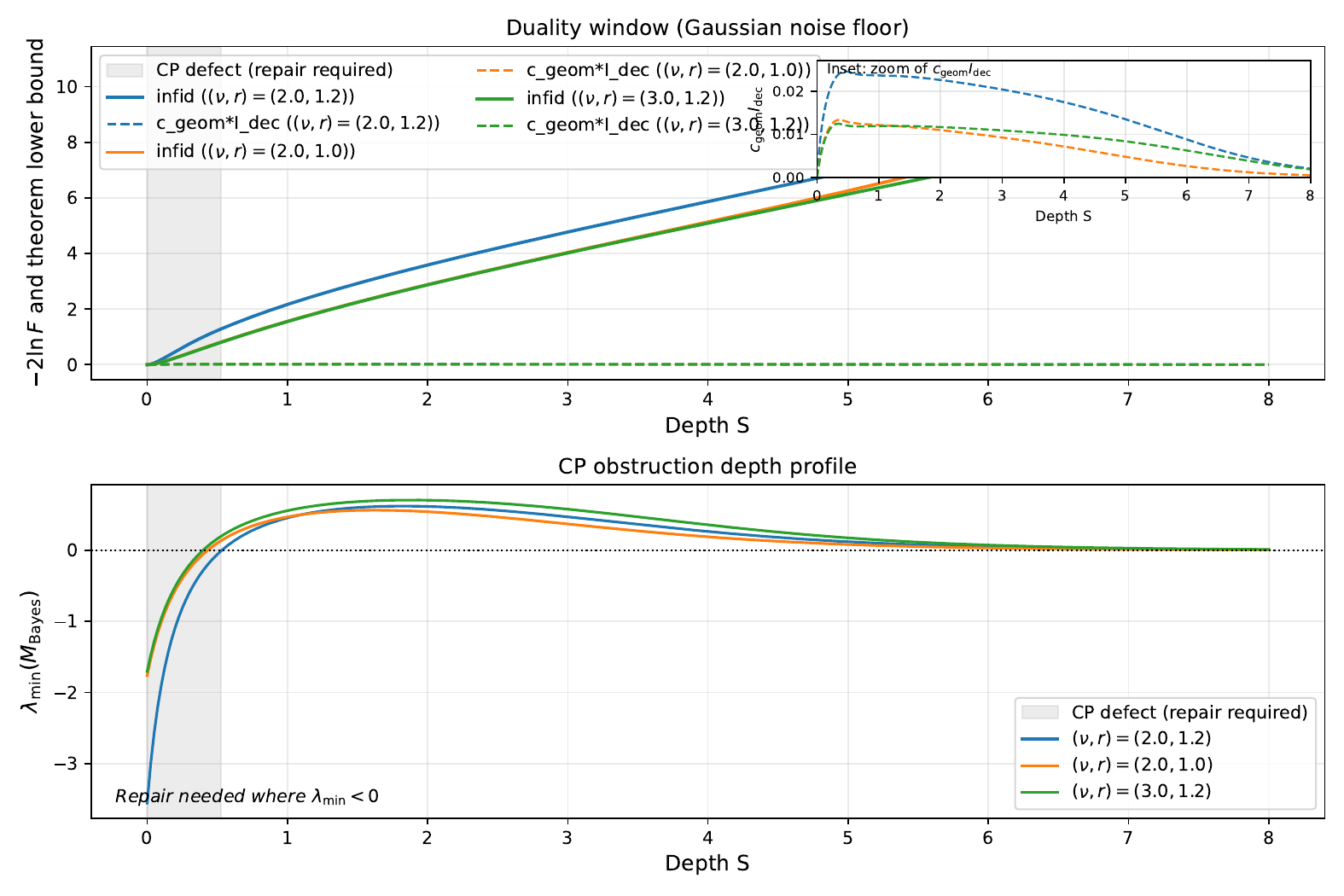}
\caption{Quantum noise floor for Gaussian decoders.
Worst-case infidelity $-2\ln F$ (solid) and the lower bound $\cgeom(\nu_{\min})\,\mathcal{I}_{\mathrm{dec}}^{\mathrm{wc}}(S)$ (dashed) versus depth $S$ for three representative HHW settings $(\nu,r)$.
The shaded band marks depths where $\lambda_{\min}(M^{\mathrm{Bayes}})<0$ (CP defect, repair required), and the inset zooms the theorem lower bound.
Details are in Supplemental Material~\cite{SM}.}
\label{fig:noisefloor}
\end{figure}

\paragraph{Discussion.}
Our results expose a structural quantum--classical gap in diffusion reversal.
Classically, score reversal is always a valid diffusion at fixed $D^{\mathrm{cl}}\succeq 0$~\cite{Anderson1982,Haussmann1986}.
Quantum mechanically, the same Wigner/Bayes score-lift does not in general correspond to a CP Gaussian generator at fixed diffusion: the HHW generator CP condition~\eqref{eq:genCP} forces additional diffusion once the reference becomes sufficiently nonclassical.
Theorem~\ref{thm:nogo} gives a sharp threshold, and Theorem~\ref{thm:floor} gives an operational fidelity floor that any repaired Gaussian reverse decoder must respect.

\emph{Quantum generative models.}
Within the Gaussian semigroup ansatz, the learned score is insufficient to specify a physical reverse channel at fixed diffusion; CP enforcement induces an irreducible noise cost quantified by~\eqref{eq:noisefloor}.
Non-Gaussian decoders (including Petz-type recovery maps) lie outside the fixed-diffusion Gaussian reverse-semigroup ansatz and may evade Gaussian-specific bounds; the present contribution is to quantify the unavoidable cost \emph{within} the Markovian Gaussian reverse-semigroup class.

\emph{Thermodynamics of quantum channels.}
The irreversibility functional $\mathcal{I}_{\mathrm{dec}}$ measures entropy production forced by CP via the quantum Fisher information and the de~Bruijn identity, connecting diffusion reversal to thermodynamic costs of physical recovery maps.

\emph{Outlook.}
Beyond the phase-covariant one-mode case, multimode channels with noncommuting squeezing directions yield richer CP-defect geometry and repaired-noise structure.
Extending sharp operational floors and identifying the tightest near-pure-state geometry remain key targets.

\bibliographystyle{apsrev4-2}
\bibliography{ref1}

\clearpage
\onecolumngrid

\setcounter{section}{0}
\setcounter{equation}{0}
\setcounter{figure}{0}
\setcounter{table}{0}
\renewcommand{\thesection}{S\arabic{section}}
\renewcommand{\theequation}{S\arabic{equation}}
\renewcommand{\thefigure}{S\arabic{figure}}
\renewcommand{\thetable}{S\arabic{table}}

\begin{center}
{\large\bf Supplemental Material}\\[0.5ex]
for\\[0.5ex]
\emph{Quantum Diffusion Models: Complete-Positivity Obstruction to Gaussian Score Reversal}\\[1ex]
Ammar Fayad
\end{center}

\paragraph*{Contents.}
This Supplemental Material contains: (i) derivation of the generator CP matrix from the HHW Gaussian CP criterion; (ii) equivalence between HHW CP and physicality of $(\Phi\otimes\id)$ on a two-mode squeezed vacuum (TMSV), yielding a concrete CP witness; (iii) Petz monotone-metric comparison yielding $\cgeom(\nu)$ on the displacement tangent sector relevant to $\Jq$; (iv) an endpoint conversion for $-2\ln F$; (v) the SDP defining the minimal CP repair diffusion; and (vi) definitions of the verification checks used in the numerical validation.

\section{Conventions and objects used in the Letter}

\paragraph*{HHW continuous-variable conventions.}
We adopt the HHW normalization: canonical operators satisfy $[Q_j,P_k]=i\delta_{jk}$ and the symplectic form is
\begin{equation}
\sympl=\bigoplus_{j=1}^n
\begin{pmatrix}
0 & 1\\
-1 & 0
\end{pmatrix}.
\end{equation}
A real covariance matrix $\cov$ is physical iff $\cov+i\sympl\succeq 0$, equivalently all symplectic eigenvalues satisfy $\nu_k\ge 1$ (vacuum $\nu_k=1$).~\cite{HHW2010}

\paragraph*{Gaussian channel parameters.}
A centered Gaussian channel $\Phi$ acts on Weyl operators by
\begin{equation}
W_\xi \ \mapsto\ W_{X^T\xi}\exp\!\left(-\half\,\xi^T Y \xi\right),
\end{equation}
and on covariances by $\cov\mapsto X\cov X^T+Y$ for real matrices $X,Y$.

\paragraph*{Gaussian semigroup generator (covariance form).}
For a (possibly time-inhomogeneous) Gaussian dynamical semigroup, the covariance evolves as
\begin{equation}
\frac{d}{dt}\cov_t=\drift_t\cov_t+\cov_t\drift_t^T+\diffu_t.
\end{equation}

\section{From HHW CP to the generator CP matrix}
\label{sec:genCP_SM}

\paragraph*{HHW CP criterion.}
A centered Gaussian channel $(X,Y)$ is completely positive iff~\cite{HHW2010}
\begin{equation}
Y \ \succeq\ i\!\left(\sympl - X\sympl X^T\right).
\label{eq:hhwcp_SM}
\end{equation}

\paragraph*{Infinitesimal expansion.}
Consider an infinitesimal step over $dt$ with $X=\id+\drift\,dt+o(dt)$ and $Y=\diffu\,dt+o(dt)$.
Then
\begin{align}
X\sympl X^T
&=(\id+\drift\,dt)\sympl(\id+\drift^T dt)+o(dt) \nonumber\\
&=\sympl+(\drift\sympl+\sympl\drift^T)\,dt+o(dt),
\end{align}
hence
\begin{equation}
i(\sympl-X\sympl X^T)=-i(\drift\sympl+\sympl\drift^T)\,dt+o(dt).
\end{equation}
Substituting into~\eqref{eq:hhwcp_SM} and dividing by $dt>0$ yields the generator condition
\begin{equation}
\boxed{\quad
\CPmat \;:=\;\diffu+i(\drift\sympl+\sympl\drift^T)\ \succeq\ 0.
\quad}
\label{eq:genCP_SM}
\end{equation}

\section{CP violation witnessed by an unphysical TMSV output (Choi/TMSV test)}
\label{sec:tmsv}

\paragraph*{TMSV covariance (HHW units).}
For one mode in each subsystem $A,B$, define the TMSV covariance with parameter $\mu>1$ by
\begin{equation}
\cov_{AB}(\mu)=
\begin{pmatrix}
\mu\,\id_2 & \sqrt{\mu^2-1}\,Z\\
\sqrt{\mu^2-1}\,Z & \mu\,\id_2
\end{pmatrix},
\qquad
Z=\mathrm{diag}(1,-1),
\label{eq:tmsv}
\end{equation}
and $\sympl_{AB}=\sympl\oplus\sympl$.

\paragraph*{Schur complement.}
Let $\Phi$ be a centered one-mode Gaussian channel $(X,Y)$ acting on subsystem $A$. The output covariance is
\begin{equation}
\cov'_{AB}(\mu)=
\begin{pmatrix}
\mu XX^T + Y & \sqrt{\mu^2-1}\,XZ\\
\sqrt{\mu^2-1}\,ZX^T & \mu\,\id_2
\end{pmatrix}.
\end{equation}
Physicality requires $\cov'_{AB}(\mu)+i\sympl_{AB}\succeq 0$.
Write
\begin{equation}
H(\mu):=\cov'_{AB}(\mu)+i\sympl_{AB}=
\begin{pmatrix}
A(\mu)+i\sympl & C(\mu)\\
C(\mu)^T & B(\mu)+i\sympl
\end{pmatrix},
\end{equation}
with $A(\mu)=\mu XX^T+Y$, $B(\mu)=\mu\,\id_2$, and $C(\mu)=\sqrt{\mu^2-1}\,XZ$.
For $\mu>1$,
\begin{equation}
(\mu\,\id_2+i\sympl)^{-1}=\frac{1}{\mu^2-1}\,(\mu\,\id_2-i\sympl),
\end{equation}
since $(\mu\id_2+i\sympl)(\mu\id_2-i\sympl)=(\mu^2-1)\id_2$.
The Schur complement criterion gives
\begin{equation}
H(\mu)\succeq 0
\ \Longleftrightarrow\
B(\mu)+i\sympl\succeq 0\ \text{ and }\ S(\mu):=A(\mu)+i\sympl - C(\mu)\,(B(\mu)+i\sympl)^{-1}C(\mu)^T\succeq 0.
\end{equation}
Using $Z\sympl Z=-\sympl$, one finds
\begin{equation}
S(\mu)=Y+i\!\left(\sympl - X\sympl X^T\right),
\end{equation}
independent of $\mu$. Therefore,
\begin{equation}
\boxed{\quad
\cov'_{AB}(\mu)+i\sympl_{AB}\succeq 0\ \ \forall \mu>1
\quad\Longleftrightarrow\quad
Y\succeq i(\sympl - X\sympl X^T).
\quad}
\label{eq:cp_tmsv_equiv}
\end{equation}

\paragraph*{Infinitesimal witness.}
For $X=\id+\drift\,dt$ and $Y=\diffu\,dt$, Eq.~\eqref{eq:cp_tmsv_equiv} gives
\begin{equation}
S(\mu)=dt\Big(\diffu+i(\drift\sympl+\sympl\drift^T)\Big)+o(dt)=dt\,\CPmat+o(dt).
\end{equation}
Thus $\CPmat\not\succeq 0$ implies violation of the uncertainty relation for $(\Phi_{dt}\otimes\id)(\mathrm{TMSV}_\mu)$ for sufficiently small $dt>0$.

\section{Reverse-time sign conventions and CP-spectrum invariance under \texorpdfstring{$\drift\mapsto-\drift$}{K->-K}}
\label{sec:sign}

\begin{lemma}[Spectrum invariance under drift sign flip]\label{lem:sign}
Let $\CPmat(\drift):=\diffu+i(\drift\sympl+\sympl\drift^T)$ for real $\drift$ and symmetric real $\diffu$.
Then $\CPmat(-\drift)=\CPmat(\drift)^{*}$, hence $\CPmat(\drift)$ and $\CPmat(-\drift)$ have identical real spectra.
\end{lemma}

\begin{proof}
Since $\diffu$ is real and $\sympl$ is real antisymmetric, $(\drift\sympl+\sympl\drift^T)$ is real, so
\[
\CPmat(-\drift)=\diffu-i(\drift\sympl+\sympl\drift^T)=\big(\diffu+i(\drift\sympl+\sympl\drift^T)\big)^*=\CPmat(\drift)^*.
\]
Complex conjugation preserves eigenvalues because $\CPmat(\drift)$ is Hermitian.
\end{proof}

\section{Petz metric comparison and the constant \texorpdfstring{$\cgeom(\nu)$}{cgeom(nu)}}
\label{sec:cgeom}

\paragraph*{Thermal spectrum.}
A one-mode gauge-invariant Gaussian state with covariance $\cov=\nu\,\id_2$ ($\nu\ge 1$) is a thermal state diagonal in the Fock basis with mean occupation $n=(\nu-1)/2$ (HHW units). Its eigenvalues satisfy
\begin{equation}
\lambda:=\frac{p_{k+1}}{p_k}=\frac{n}{n+1}=\frac{\nu-1}{\nu+1}\in[0,1).
\label{eq:lambda_ratio}
\end{equation}

\paragraph*{Monotone metrics.}
In an eigenbasis of $\rho$, Petz monotone metrics admit the form~\cite{Petz1996,mosonyi2011quantum}
\begin{equation}
g_f^\rho(X,X)=\sum_{i,j}\frac{|X_{ij}|^2}{p_j\,f(p_i/p_j)}.
\label{eq:petz_form}
\end{equation}
Two choices are used:
(i) SLD/Bures with $f_{\mathrm{SLD}}(t)=(1+t)/2$ and $g^{\Bures}=\frac14 g^{\mathrm{SLD}}$; and
(ii) BKM with $f_{\mathrm{BKM}}(t)=(t-1)/\ln t$.
Define
\begin{equation}
r(t)=\frac{g^{\Bures}}{g^{\BKM}}
=\frac14\,\frac{f_{\mathrm{BKM}}(t)}{f_{\mathrm{SLD}}(t)}
=\frac12\,\frac{t-1}{(t+1)\ln t}.
\label{eq:rdef}
\end{equation}

\begin{lemma}[Adjacent-ratio restriction for displacement tangents]\label{lem:adjacent}
Let $\rho_\nu$ be a one-mode thermal state. Any displacement tangent operator is a real linear combination of $Q$ and $P$, equivalently of $a+a^\dagger$ and $(a-a^\dagger)/i$. In the Fock basis, these have matrix elements only between adjacent levels. Consequently, in~\eqref{eq:petz_form} restricted to displacement tangents, the only eigenvalue ratios that appear are $t=\lambda$ and $t=\lambda^{-1}$.
\end{lemma}

\begin{proof}
In the Fock basis $\{|k\rangle\}_{k\ge 0}$, $a|k\rangle=\sqrt{k}\,|k-1\rangle$ and $a^\dagger|k\rangle=\sqrt{k+1}\,|k+1\rangle$, hence $Q$ and $P$ are one-step banded. For $\rho_\nu=\sum_k p_k |k\rangle\langle k|$ with geometric spectrum $p_{k+1}/p_k=\lambda$, only ratios $p_{k+1}/p_k=\lambda$ and $p_k/p_{k+1}=\lambda^{-1}$ occur in nonzero off-diagonal terms.
\end{proof}

\paragraph*{Evaluation of $\cgeom(\nu)$.}
Since $r(t)=r(1/t)$, the displacement-sector ratio evaluates to $r(\lambda)$:
\begin{align}
r(\lambda)
&=\frac12\,\frac{\lambda-1}{(\lambda+1)\ln\lambda}
=\frac12\,\frac{1-\lambda}{(1+\lambda)\ln(1/\lambda)}.
\end{align}
Using~\eqref{eq:lambda_ratio},
\begin{equation}
\frac{1-\lambda}{1+\lambda}=\frac{1}{\nu},
\qquad
\ln(1/\lambda)=\ln\!\left(\frac{\nu+1}{\nu-1}\right).
\end{equation}
Therefore, on the displacement tangent sector of thermal (gauge-invariant) states,
\begin{equation}
\boxed{\quad
g^{\Bures}\ \succeq\ \cgeom(\nu)\,g^{\BKM},
\qquad
\cgeom(\nu)=\frac{1}{2\nu\ln\!\left(\frac{\nu+1}{\nu-1}\right)}\in(0,1/4].
\quad}
\label{eq:cgeom_SM}
\end{equation}
For multimode gauge-invariant Williamson normal modes with symplectic eigenvalues $\nu_k\ge\nu_{\min}>1$, the same bound holds with $\cgeom(\nu_{\min})$.

\paragraph*{Near-pure-state degeneration.}
As $\nu\to 1^+$, $\ln\frac{\nu+1}{\nu-1}\to\infty$, hence $\cgeom(\nu)\to 0$, motivating the Letter’s explicit $\nu_{\min}>1$ assumption.

\section{Endpoint fidelity bound}
\label{sec:endpoint}

\paragraph*{Endpoint conversion.}
Let $\mathcal{A}(\rho,\sigma):=\arccos\!\sqrt{F(\rho,\sigma)}$ be the Bures angle. Since
$F=\cos^2\mathcal{A}$ and $-2\ln(\cos^2 x)\ge 2x^2$ for $x\in[0,\pi/2)$, one has
\begin{equation}
-2\ln F(\rho,\sigma)\ \ge\ 2\,\mathcal{A}(\rho,\sigma)^2.
\label{eq:logF_angle}
\end{equation}

\paragraph*{Scope of use.}
The bound in the Letter is obtained by combining the exact local de~Bruijn identity (BKM/DQFI cost density) with the uniform monotone-metric dominance~\eqref{eq:cgeom_SM} on the same displacement tangent sector.

\section{Minimal CP repair as a convex SDP}
\label{sec:sdp}

At each depth/time $s$, the Bayes/score-lift produces a candidate generator CP matrix $\CPmat^{\mathrm{Bayes}}(s)$.
A repaired Gaussian reverse decoder adds diffusion $\Delta \diffu_{\mathrm{qu}}(s)\succeq 0$ such that
\begin{equation}
\CPmat^{\mathrm{Bayes}}(s)+\Delta \diffu_{\mathrm{qu}}(s)\ \succeq\ 0.
\end{equation}
The pointwise repair is defined by the convex program
\begin{equation}
\boxed{\quad
\Delta \diffu_{\mathrm{qu}}^\star(s)
\in
\arg\min_{\substack{\Delta \diffu=\Delta \diffu^T\succeq 0\\ \CPmat^{\mathrm{Bayes}}(s)+\Delta \diffu\succeq 0}}
\tr\!\big(\Delta \diffu\;\Jq[\tau_s]\big).
\quad}
\label{eq:sdp_SM}
\end{equation}

\end{document}